\documentclass[conference]{IEEEtran}
\IEEEoverridecommandlockouts

\usepackage{amsmath,amssymb,amsfonts,amsthm,mathtools}
\usepackage{stmaryrd}
\usepackage{booktabs}
\usepackage{bm}
\usepackage{graphicx}
\usepackage{cite}
\usepackage{hyperref}

\usepackage{tikz}
\usetikzlibrary{arrows.meta,positioning,calc,fit,backgrounds,decorations.pathreplacing,shapes.geometric,decorations.markings,angles,quotes}

\newtheorem{theorem}{Theorem}
\newtheorem{lemma}{Lemma}
\newtheorem{proposition}{Proposition}

\newtheorem{definition}{Definition}

\DeclareMathOperator{\rank}{rank}

\title{
Bias-Aware BP Decoding of Quantum Codes\\ via Directional Degeneracy
}

\author{
\IEEEauthorblockN{Mohammad Rowshan}
\IEEEauthorblockA{
University of New South Wales (UNSW)\\
Sydney, Australia\\
Email: mrowshan@ieee.org}
}

\begin{document}
\maketitle

\begin{abstract}
We study \emph{directionally informed belief propagation (BP) decoding} for quantum CSS codes, where anisotropic Tanner-graph structure and biased noise concentrate degeneracy along preferred directions. We formalize this by placing orientation weights on Tanner-graph edges, aggregating them into per-qubit directional weights, and defining a \emph{directional degeneracy enumerator} that summarizes how degeneracy concentrates along those directions. A single bias parameter~$\beta$ maps these weights into site-dependent log-likelihood ratios (LLRs), yielding anisotropic priors that plug directly into standard BP$\rightarrow$OSD decoders without changing the code construction. We derive bounds relating directional and Hamming distances, upper bound the number of degenerate error classes per syndrome as a function of distance, rate, and directional bias, and give a MacWilliams-type expression for the directional enumerator. Finite-length simulations under code-capacity noise show significant logical error-rate reductions---often an order of magnitude at moderate physical error rates---confirming that modest anisotropy is a simple and effective route to hardware-aware decoding gains.
\end{abstract}

\section{Introduction}

Quantum LDPC (qLDPC) codes combine sparse parity checks with families of constant rate and linear (or near-linear) distance~\cite{tillich_zemor_2014,leverrier_tillich_zemor_2015,evra_kaufman_zemor_2020,panteleev_kalachev_asymp_2022,leverrier_zemor_tanner_2022,breuckmann_eberhardt_2021}. Lifted-product codes~\cite{panteleev_kalachev_asymp_2022} and quantum Tanner codes~\cite{leverrier_zemor_tanner_2022} provide asymptotically good families with efficient decoders, yet practical decoding at finite blocklength is complicated by short cycles and \emph{degeneracy}, the fact that many distinct errors share a syndrome and act identically on the codespace~\cite{poulin_chung_2008,panteleev_kalachev_deg_2021}.

Physical noise processes are typically anisotropic: dephasing errors can dominate bit flips and mixed errors, and hardware layouts often favor certain coupling directions. Bias-tailored surface codes and related CSS constructions show that aligning stabilizers and decoders with such bias can produce significant threshold and overhead improvements~\cite{tuckett_biased_prx_2019,bonilla_xzzx_2021,roffe_bias_tailored_qldpc_2023}. Most of this literature either changes the \emph{code} to match a given biased channel, or uses decoders whose priors do not explicitly encode directional structure in the Tanner graph.

This work focuses on the decoder side and keeps the code fixed. We annotate the Tanner graphs of a CSS code with directional edge weights that reflect layout, scheduling, or calibrated anisotropy. Summing these edge weights gives per-qubit directional weights and a weighted Hamming cost, which in turn define a directional degeneracy enumerator. A single parameter $\beta$ converts these directional weights into site-dependent LLRs for a standard BP$\rightarrow$OSD decoder. The same $\beta$ controls a family of partition functions that count and reweight degenerate error classes, yielding both analytical insight and a practical tuning knob.

At a high level, the contributions are as follows. First, we introduce directional edge weights and the resulting per-qubit weights, and use them to define a directional degeneracy metric and enumerator. Second, we derive bounds that compare directional and Hamming distances, and prove a directionality-induced degeneracy bound that makes explicit how directional bias reduces the number of degenerate error classes per syndrome. Third, we present an anisotropic BP+OSD decoder that uses the directional weights through a single bias parameter, and we show through finite-length simulations that this anisotropy yields substantial logical error-rate reductions without altering the underlying code.

\section{Notations and Preliminaries}
\label{sec:prelim}

We consider CSS stabilizer codes specified by binary parity-check matrices
\[
H_X\in\mathbb{F}_2^{m_X\times n},\qquad H_Z\in\mathbb{F}_2^{m_Z\times n},
\]
satisfying the commutation condition $H_XH_Z^{\top}=0$. Their row spaces
\[
S_X:=\operatorname{rowsp}(H_X),\qquad S_Z:=\operatorname{rowsp}(H_Z)
\]
encode the supports of $X$- and $Z$-stabilizer generators. Associated classical codes are
\[
C_X:=\ker(H_X),\qquad C_Z:=\ker(H_Z),
\]
with $C_Z\subseteq C_X^\perp$ and $C_X\subseteq C_Z^\perp$ under the standard inner product. The number of logical qubits is
\[
k \;=\; n-\rank(H_X)-\rank(H_Z).
\]

Throughout we use the CSS picture: $X$- and $Z$-errors are decoded separately. For a fixed $Z$-syndrome $s_Z$, let $e_0$ be a solution of $H_Z e=s_Z$. All solutions are of the form $e_0+C_Z$. Two such solutions are indistinguishable if they differ by an $X$-stabilizer, i.e., if their difference lies in $S_X$.

\begin{definition}[Degeneracy classes]
For a fixed $Z$-syndrome $s_Z$, the set of degenerate error classes is the quotient
\begin{equation}
\mathcal{D}_X(s_Z)\;\cong\; (e_0+C_Z)\big/ S_X.
\label{eq:deg-classes}
\end{equation}
Its size is $|\mathcal{D}_X(s_Z)|=2^k$ and does not depend on the syndrome.
\end{definition}

An analogous definition applies on the $Z$ side for $X$-syndromes. The objects in $\mathcal{D}_X(s_Z)$ index distinct logical $X$-actions compatible with the observed $Z$-syndrome. The directional framework developed below reweights these classes based on their alignment with a preferred direction.

\section{Directional Weights and Degeneracy Enumerator}
\label{sec:dir-enum}

\subsection{Directional annotation of the Tanner graph}

We view the CSS code through its Tanner graphs. For qubit $i$, let $N_X(i)$ be the set of adjacent $X$-checks and $N_Z(i)$ the set of adjacent $Z$-checks in bipartite graphs of $H_X$ and $H_Z$.

\begin{definition}[Directional annotation]
A \emph{directionally annotated} CSS code consists of $(H_X,H_Z)$ together with nonnegative edge weight matrices
\[
D_X\in\mathbb{R}_{\ge 0}^{n\times m_X},\qquad
D_Z\in\mathbb{R}_{\ge 0}^{n\times m_Z},
\]
supported on the edges of the $X$- and $Z$-Tanner graphs, respectively. A larger entry $D_X(i,j)$ or $D_Z(i,j)$ indicates that qubit $i$ couples more strongly to check $j$ along a preferred direction, for instance due to geometry, scheduling, or hardware-specific anisotropy (e.g., specific behavior of noise).
\end{definition}

Summing the edge weights incident to qubit $i$ produces a per-qubit directional weight
\begin{equation}
w_i \;:=\; \sum_{\mathclap{j\in N_X(i)}} D_X(i,j)+\sum_{\mathclap{j\in N_Z(i)}} D_Z(i,j),
\quad \bm w=(w_1,\dots,w_n).
\label{eq:per-qubit}
\end{equation}
So $w_i$ says: how strongly qubit  $i$ is involved in ``directionally important" edges overall. 
For a binary error indicator $E\in\{0,1\}^n$, the corresponding directional cost is then
\begin{equation}
\Delta_{\bm w}(E)\;=\;\sum_{i=1}^n w_i E_i
\;=\;\langle \bm w,E\rangle.
\label{eq:weighted-hamming}
\end{equation}

\begin{figure}[t]
  \centering
  \tikzset{>=Latex}
\begin{tikzpicture}[x=1.2cm,y=1.2cm,>=Stealth, font=\small]

\tikzset{
  qubit/.style  = {circle, draw=black, thick, minimum size=7mm, inner sep=0pt, fill=white},
  xcheck/.style = {rectangle, draw=red!70!black, thick, minimum width=7mm, minimum height=7mm, fill=red!8},
  zcheck/.style = {diamond, draw=blue!70!black, thick, minimum size=7.5mm, inner sep=0pt, fill=blue!8},
  DX/.style     = {draw=red!70!black,  line width=#1},
  DZ/.style     = {draw=blue!70!black, line width=#1}
}

\draw[very thick, -{Stealth[length=3mm,width=3mm]}] (-0.2,2.6) -- (3.8,2.6)
  node[midway,above] {$\hat d$ (preferred direction)};

\node[qubit] (q1) at (0, 2) {$q_1$};
\node[qubit] (q2) at (0, 1) {$q_2$};
\node[qubit] (q3) at (0, 0) {$q_3$};

\node[xcheck] (X1) at (3, 2) {$X_1$};
\node[zcheck] (Z1) at (3, 0) {$Z_1$};

\draw[DX=1.6pt] (q1) -- node[above,sloped,red!70!black] {\scriptsize $D_X(1,1)$} (X1);
\draw[DX=0.9pt] (q2) -- node[below,sloped,red!70!black] {\scriptsize $D_X(2,1)$} (X1);

\draw[DZ=1.6pt] (q2) -- node[above,sloped,blue!70!black] {\scriptsize $D_Z(2,1)$} (Z1);
\draw[DZ=0.9pt] (q3) -- node[below,sloped,blue!70!black] {\scriptsize $D_Z(3,1)$} (Z1);

\end{tikzpicture}
  \caption{Directionally annotated Tanner graph. Edge thickness encodes the magnitude of orientation weights $D_X,D_Z$. Summing over incident edges yields per-qubit directional weights $w_i$ that define the cost $\Delta_{\bm w}(E)$.}
  \label{fig:dir-tanner}
\end{figure}
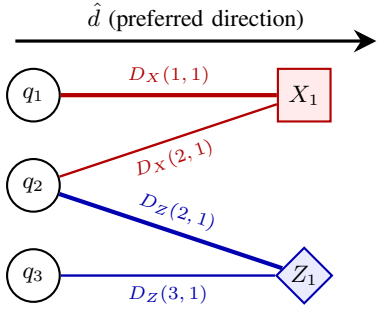
That’s just a weighted Hamming weight: each flipped qubit 
$i$ contributes $w_i$ instead of ``1". Errors on ``directionally important" qubits are thus counted as more expensive.

Figure~\ref{fig:dir-tanner} shows a toy Tanner graph with edge weights. In practice, the entries of $D_X,D_Z$ can be derived from geometry, readout and gate scheduling, or from hardware calibration data, but the theory below requires only that they be nonnegative.

\begin{proposition}[Edge-to-qubit reduction]
\label{prop:perqubit}
For an error indicator $E\in\{0,1\}^n$, define
\[
\Delta_D(E)
:= \sum_{i : E_i = 1} w_i, 
\]
that is, the sum of directional edge weights incident on the support of $E$.
Then, $\Delta_D(E)=\Delta_{\bm w}(E)$ with $w_i$ given by~\eqref{eq:per-qubit}.
\end{proposition}

\begin{proof}
Using $E_i\in\{0,1\}$, 
$
\Delta_D(E)
= \sum_{i=1}^n w_i E_i
= \Delta_{\bm w}(E).
$
\end{proof}

Thus, the directional annotation can be reduced to a per-qubit vector $\bm w$ and a linear functional $\Delta_{\bm w}$, which are the only ingredients used in the enumerator and in the decoder.

\subsection{Class scores and directional enumerator}
For a fixed $Z$-syndrome $s_Z$, we have grouped all compatible $X$-errors into degeneracy classes $\mathcal{D}_X(s_Z)$. In decoding, we ultimately choose between \emph{classes}, not individual representatives, so we need a single ``directional cost'' for each class built from $\Delta_{\bm w}$. 
Because $\Delta_{\bm w}$ is not invariant under adding stabilizers, classes require their own scores.

\begin{definition}[Class score]
For a fixed $Z$-syndrome $s_Z$ and a class $[e]\in\mathcal{D}_X(s_Z)$, the \emph{class score} is
\begin{equation}
\Delta_*([e])\;:=\;\min_{u\in S_X} \Delta_{\bm w}(e+u),
\label{eq:class-score}
\end{equation}
which is the cost of the least-cost representative, or weighted coset leader.
\end{definition}

\begin{definition}[Directional degeneracy enumerator]
For a fixed $Z$-syndrome $s_Z$ and weights $\bm w$, the \emph{directional degeneracy enumerator} is
\begin{equation}
\Gamma_X(s_Z;\beta)\;:=\;\sum_{[e]\in\mathcal{D}_X(s_Z)} \exp\!\big(-\beta\,\Delta_*([e])\big),
\qquad \beta\ge 0.
\label{eq:enum}
\end{equation}
\end{definition}

At $\beta=0$ the enumerator equals $2^k$; its first derivative gives the mean class score and its second derivative the score variance. For $\beta>0$, $\Gamma_X(s_Z;\beta)$ is dominated by classes with small $\Delta_*([e])$, and its decay in $\beta$ measures how many such low-cost classes exist: fast decay means few, slow decay many. Thus $\Gamma_X$ compactly summarizes how degeneracy distributes across directions under the per-qubit penalties $\bm w$.

\begin{lemma}[Tail bound for low-cost classes]
\label{lem:tail}
Let $M_{\le t}$ be the number of classes with $\Delta_*([e])\le t$. For any $\beta>0$ and $t>0$,
\[
M_{\le t} \;\le\; e^{\beta t}\,\Gamma_X(s_Z;\beta).
\]
\end{lemma}
\begin{proof}
Restricting the sum in \eqref{eq:enum} to classes with cost at most $t$ yields $\Gamma_X(s_Z;\beta)\ge M_{\le t} e^{-\beta t}$.
\end{proof}

Lemma~\ref{lem:tail} shows that, for fixed $\beta$, a small enumerator forces the number of low-cost classes to be small. Directional weighting thus provides a handle for thinning out competitive low-cost classes, which is what we want for decoding: with fewer nearly equivalent low-cost options, the decoder has fewer opportunities to confuse logical errors.

\section{Analytical Properties}
\label{sec:analysis}

The directional enumerator admits useful inequalities and dual representations. We first compare directional and Hamming distances and then give a MacWilliams-type form for a global enumerator.

\subsection{Directional distances}

Let $d$ be the (Hamming) code distance and $d_{\mathrm{S}}$ the minimum Hamming weight among nontrivial stabilizers in $S_X\cup S_Z$. Let $w_{\min}=\min_i w_i$ and $w_{\max}=\max_i w_i$.

\begin{proposition}[Comparison with Hamming distance]
\label{prop:dir-vs-hamming}
Let $d_{\bm{w}}^{\mathrm{S}}$ be the minimum directional cost among nontrivial stabilizers and $d_{\bm{w}}^{\mathrm{L}}$ the minimum directional cost among nontrivial logical operators. Then
\[
w_{\min}\,d_{\mathrm{S}} \le d_{\bm{w}}^{\mathrm{S}}\le w_{\max}\,d_{\mathrm{S}},
\qquad
w_{\min}\,d \le d_{\bm{w}}^{\mathrm{L}}\le w_{\max}\,d.
\]
\end{proposition}
\begin{proof}
For any nonzero vector $v$, we have $w_{\min}\mathrm{wt}(v)\le \sum_i w_i v_i \le w_{\max}\mathrm{wt}(v)$. Applying this to the Hamming-minimizing stabilizers and logicals gives the claim.
\end{proof}

Directional distances therefore inherit scaling from classical distances. Mild anisotropy ($w_{\min}\approx w_{\max}$) preserves distances up to a small factor; strong anisotropy can alter them significantly if logical and stabilizer supports are aligned with low- or high-weight directions.

\subsection{Directionality-induced degeneracy bound}

We bound how many degenerate error classes remain admissible once a directional model is imposed. Let $\llbracket n,k,d_{\min}\rrbracket$ be a CSS code with rate $R=k/n$ and minimum distance $d_{\min}$. Suppose the directional model filters classes according to an admissibility rule. For example, one may require that an admissible class contain a representative whose support lies mostly along the favored direction according to $D_X$.

For each syndrome $s_Z$, let $\mathcal{D}_\delta(s_Z)\subseteq\mathcal{D}_X(s_Z)$ be the subset of classes that are admissible under the directional model. Define the worst-case concentration factor
\[
f(\delta_{\max},R)\;:=\;\sup_{s_Z}\frac{|\mathcal{D}_\delta(s_Z)|}{|\mathcal{D}_X(s_Z)|},
\]
where $\delta_{\max}\ge 0$ is any scalar that summarizes the maximum directional bias, for instance $\delta_{\max}=\max_{i,j}D_X(i,j)$. Since directionality only removes classes, $f(\delta_{\max},R)\le 1$ and $f(0,R)=1$.

\begin{lemma}[Directionality-induced degeneracy bound]
\label{lem:dir-deg}
For every syndrome $s_Z$,
\[
|\mathcal{D}_\delta(s_Z)| \;\le\; 2^{k}\,f(\delta_{\max},R)
\;\le\; 2^{\,n-2d_{\min}+2}\,f(\delta_{\max},R).
\]
\end{lemma}
\begin{proof}
The isotropic degeneracy count is $|\mathcal{D}_X(s_Z)|=2^{k}$; admissibility only removes classes, giving $|\mathcal{D}_\delta(s_Z)|\le 2^{k}f(\delta_{\max},R)$. The quantum Singleton bound~\cite{KnillLaflamme1997} implies $k\le n-2d_{\min}+2$, hence $2^{k}\le 2^{n-2d_{\min}+2}$.
\end{proof}

Lemma~\ref{lem:dir-deg} makes explicit how directional structure can reduce the effective degeneracy of the code. The concentration factor $f(\delta_{\max},R)$ depends on the directional rule and on the code family; it is monotone nonincreasing in the strength of the bias. In examples where directional structure correlates with logical operators, one expects $f(\delta_{\max},R)$ to be significantly smaller than $1$ for moderate $\delta_{\max}$.

\subsection{Global enumerator and MacWilliams representation}

Let $C:=C_X\cap C_Z$ and consider the global directional enumerator
\begin{equation}
\Gamma(\bm w;\alpha)\;:=\;\sum_{v\in C} e^{\alpha \langle \bm w,v\rangle},
\qquad \alpha>0.
\label{eq:global-enum}
\end{equation}
This is the specialization at $(x_i,y_i)=(1,e^{\alpha w_i})$ of the per-coordinate complete weight enumerator
\[
W_C(\{x_i,y_i\}_{i=1}^n)\;:=\;\sum_{v\in C}\prod_{i=1}^n x_i^{1-v_i}y_i^{v_i}.
\]

\begin{theorem}[MacWilliams-type specialization]
\label{thm:macw}
Let $C^\perp$ be the dual of $C$ under the standard inner product. Then
\[
\Gamma(\bm w;\alpha)
= W_C(\{1,e^{\alpha w_i}\})
= \frac{1}{|C^\perp|}\sum_{u\in C^\perp}\prod_{i=1}^n\big(1+(-1)^{u_i}e^{\alpha w_i}\big).
\]
\end{theorem}
\begin{proof}
The first equality is the specialization $x_i=1$, $y_i=e^{\alpha w_i}$. The second is the MacWilliams identity for the complete weight enumerator over $\mathbb{F}_2$~\cite{rains_quantum_weight_1998}, which is a discrete Fourier transform on $\{0,1\}^n$.
\end{proof}

Theorem~\ref{thm:macw} connects directional degeneracy to the dual code and shows that the dependence on $\bm w$ factorizes over coordinates in the dual domain. This form is useful when $C^\perp$ has a structured description, such as in hypergraph-product or Tanner constructions.

The logarithm of $\Gamma$ has convenient calculus identities. Writing $\pi_{\alpha,\bm w}(v)\propto e^{\alpha\langle \bm w,v\rangle}$ for the Gibbs distribution on $C$, we have
\begin{equation}
\frac{\partial}{\partial w_i}\log \Gamma(\bm w;\alpha)
= \alpha\,\mathbb{E}_{\pi_{\alpha,\bm w}}[v_i].
\label{eq:grad}
\end{equation}
Thus increasing $w_i$ suppresses coordinates with large expected occupancy. The function $\log\Gamma(\bm w;\alpha)$ is convex in $\bm w$ since it is the log of a sum of exponentials of linear forms. Small $\ell_\infty$ perturbations of $\bm w$ change $\log\Gamma$ by at most $O(\alpha n\|\delta\|_\infty)$; a simple argument using $|\langle \delta,v\rangle|\le n\|\delta\|_\infty$ and bounding the ratio of exponentials yields a global Lipschitz constant.

\section{Anisotropic Decoding}
\label{sec:decoder}

We now describe how the directional weights are used in a practical BP+OSD decoder. The underlying code is unchanged; only the prior and the scoring function are modified.

\subsection{Site-dependent priors from directional weights}

Consider a memoryless binary-input channel (for one CSS side) with independent error probabilities $p_i$ on each qubit. Under independence, the negative log-likelihood of an error pattern $E\in\{0,1\}^n$ is
\[
-\log\Pr(E) = \sum_i E_i \log\frac{p_i}{1-p_i} + \text{const}
= \sum_i w_i^{\mathrm{MAP}} E_i + \text{const},
\]
with $w_i^{\mathrm{MAP}}=\log\frac{1-p_i}{p_i}$. Thus, MAP decoding corresponds to weighted min-sum with weights $w_i^{\mathrm{MAP}}$.

The directional weights $\bm w$ from~\eqref{eq:per-qubit} provide a natural way to parametrize $p_i$. A convenient one-parameter family is
\begin{equation}
p_i(\beta)\;=\; \frac{p_0\,e^{\beta w_i}}{\frac{1}{n}\sum_{j=1}^n e^{\beta w_j}},
\qquad
\ell_i(\beta)=\log\frac{1-p_i(\beta)}{p_i(\beta)},
\label{eq:tilting}
\end{equation}
where $p_0$ is a target average physical error rate and $\beta\ge 0$ is a directional strength parameter. For $\beta=0$ this reduces to the isotropic prior $p_i=p_0$. Increasing $\beta$ makes errors more likely on qubits with larger $w_i$.

The LLRs $\ell_i(\beta)$ are passed to a standard belief-propagation decoder on the Tanner graph of $H_X$ or $H_Z$. From the channel perspective, \eqref{eq:tilting} is an exponential tilt of a baseline isotropic distribution, and the same $\beta$ also appears in the directional enumerator through the weights.

\subsection{MAP on trees and coset posteriors}

When the factor graph of $H$ is cycle-free, belief propagation is exact.


\begin{theorem}[MAP on trees equals weighted min-sum]
\label{thm:map-tree}
Let $H\in\{0,1\}^{m\times n}$ be a parity-check matrix whose factor graph is a forest. Suppose the physical noise on each qubit is independent with error probabilities $p_i\in(0,1/2)$ and define
\[
w_i \;:=\; \log\frac{1-p_i}{p_i}.
\]
For a given syndrome $s$, let $\mathcal{S}(s):=\{E\in\{0,1\}^n : H E = s\}$ be the affine solution set. If each syndrome $s$ identifies a unique coset in $\mathcal{S}(s)$ (so that the MAP solution is unique up to stabilizers), then the MAP estimate
\[
E^\star \;\in\; \arg\max_{E\in\mathcal{S}(s)} \Pr(E\mid s)
\]
coincides with
\[
E^\star \;=\; \arg\min_{E\in\mathcal{S}(s)}\;\sum_{i=1}^n w_i E_i,
\]
and min-sum belief propagation on the factor graph of $H$ with unary costs $w_i$ returns $E^\star$.
\end{theorem}

On loopy graphs, BP is approximate, but the structure of coset posteriors remains clear.

\begin{proposition}[Coset posterior]
\label{prop:coset-posterior}
For independent noise with probabilities $p_i$ and weights $w_i$, the posterior probability of a coset $\mathcal{C}\subseteq\{E:HE=s\}$ given syndrome $s$ satisfies
\[
\Pr(\mathcal{C}\mid s)\;\propto\;
\sum_{v\in \mathcal{C}} \exp\!\Big(-\sum_{i=1}^n w_i v_i\Big).
\]
\end{proposition}
\begin{proof}
Conditioning on $HE=s$ does not change the relative likelihoods within the solution set. The likelihood of $v$ is proportional to $\exp\big(-\sum_i w_i v_i\big)$; the posterior of the coset is obtained by summing over its representatives.
\end{proof}

Thus, degeneracy enters through a multiplicity-weighted partition function on each coset, controlled by $w_i$. Changing $\beta$ in~\eqref{eq:tilting} modifies these partition functions and shifts which cosets have highest posterior mass.

\subsection{Anisotropic BP+OSD decoder}

In practice, we decode each CSS side as follows. The directional model provides $D_X,D_Z$ and per-qubit weights $\bm w$. For a given physical error rate $p_0$ and bias parameter $\beta$, we form $p_i(\beta)$ and LLRs $\ell_i(\beta)$ from~\eqref{eq:tilting}. These LLRs initialize a min-sum BP decoder, which is run for a number of iterations, and the tentative estimate is further refined by ordered-statistics decoding (OSD) of moderate order~\cite{panteleev_kalachev_deg_2021}. The same directional weights are used to rank candidates and break ties: among syndrome-consistent patterns, the decoder prefers smaller values of $\Delta_{\bm w}(E)$. 
The key point is that the code, the Tanner graph, and the OSD implementation are unchanged. The only difference between isotropic and anisotropic decoding is in how LLRs and candidate scores are computed from the directional weights and the scalar parameter $\beta$.

\section{Numerical Results}
\label{sec:numerics}


This section illustrates how anisotropic priors affect finite-length performance, via directional weights $w_i$ and a scalar $\beta$. The same machinery can bias Tanner-graph edges in many other ways, e.g.\ stripwise piecewise-constant $w_i$ for readout rows, checkerboard or layer-dependent weights modelling interleaved hardware sublattices, or radial gradients centred on a particularly noisy region of the device. 

Note that anisotropic/directional priors are not a silver bullet: for codes whose degeneracy is essentially isotropic (e.g.\ ensemble-style qLDPCs with no meaningful geometric embedding) or for hardware where the noise is close to i.i.d., a one-dimensional anisotropic tilt may offer little or no benefit and can even degrade performance if the assumed orientation is badly misaligned with the true noise. More generally, when bias arises from complicated cross-qubit correlations rather than a smooth gradient (e.g.\ strongly gate- or chip-layout–dependent error patterns), richer hardware models are likely needed beyond the simple directional field considered here.

\subsection{Orientation--based directional priors} 
\label{sec:ne3n-toric-orientation}

We now specialise the directional framework of Section~\ref{sec:decoder}
(weights $w_i$, tilted probabilities $p_i(\beta)$, and LLRs $\ell_i(\beta)$)
to two geometries: the planar $\llbracket 36,4\rrbracket$ NE3N code and the
$\llbracket 2L^2,2,L\rrbracket$ toric code ($L=9$).

\paragraph{NE3N planar code.}
The NE3N code \cite{geher2025directional} is realised on an $18\times 4$ rectangular lattice with a bipartite data/ancilla layout, giving each physical qubit $i$ integer coordinates $(x_i,y_i)\in\mathbb{Z}^2$.  To model a
horizontal hardware anisotropy (e.g.\ control lines running left–to–right),
we choose the scalar coordinate
\begin{equation}\label{eq:coord}
  c_i := x_i
\end{equation}
in the general construction of Section~\ref{sec:decoder}, and obtain
directional weights $w_i$ by standardising $\{c_i\}$ across all data qubits, to obtain
dimensionless directional weights
\begin{equation}\label{eq:mean_sigma}
  \bar c \;:=\; \frac{1}{n}\sum_{i=1}^n c_i,\qquad
  \sigma_c^2 \;:=\; \frac{1}{n-1}\sum_{i=1}^n (c_i-\bar c)^2,
\end{equation}
\begin{equation}\label{eq:qbit_wt}
  w_i \;:=\; \frac{c_i - \bar c}{\sigma_c},
\end{equation}
Qubits near the right edge then have $w_i>0$, those near the left edge have
$w_i<0$, with a smooth gradient in between.  

\paragraph{Toric code.}
For the toric code, we consider an $L\times L$ square lattice with periodic
boundary conditions, with one qubit on each horizontal and vertical edge.
Horizontal edges are labelled
\[
  \mathrm{idx}_h(x,y) = yL + x,\quad
  \mathrm{coord}\big(\mathrm{idx}_h(x,y)\big) = (2x,2y),
\]
and vertical edges
\[
  \mathrm{idx}_v(x,y) = L^2 + yL + x,\quad
  \mathrm{coord}\big(\mathrm{idx}_v(x,y)\big) = (2x+1,2y+1),
\]
for $x,y\in\{0,\dots,L-1\}$ (with arithmetic modulo $L$).  This embeds all
$2L^2$ qubits at integer coordinates $(x_i,y_i)\in\{0,\dots,2L-1\}^2$ on two
interleaved checkerboard sublattices.  To emulate a hardware gradient along one spatial axis, we again take
\(
  c_i := x_i
\)
and form $w_i$ by standardising $\{c_i\}$ as before.  Qubits at larger
$x$–coordinate are thus “downstream’’ (positive $w_i$), and errors become more
likely there as $\beta$ increases.  

\paragraph{Alternative directional fields.}
The same mechanism can be used to define $w_i$ in other hardware and code
geometries.  For example:
\begin{itemize}
  \item \emph{Vertical orientation:} set $c_i=y_i$ instead of $x_i$, so that
        directionality favours rows rather than columns.
  \item \emph{Strip-wise weights:} choose a subset of ``favoured'' columns
        $\mathcal{C}\subset\{0,\dots,L_x-1\}$ and define
        $w_i = +w_0$ if $x_i\in\mathcal{C}$ and $w_i=-w_0$ otherwise, for
        some fixed contrast $w_0>0$.
  \item \emph{Radial gradient:} for codes embedded in a disc or annulus,
        define $c_i$ as the radial coordinate and apply the same
        standardisation and tilt as in~\eqref{eq:mean_sigma}–\eqref{eq:qbit_wt}
        to favour qubits near or away from a boundary.
\end{itemize}
All of these choices are compatible with the directional enumerator and
anisotropic decoding framework developed in Sections~\ref{sec:dir-enum}
and~\ref{sec:decoder}; the only change is in how the per-qubit weights
$w_i$ are instantiated for a given code family and hardware noise model.

\subsection{Logical error-rate behavior}

\begin{figure}
    \centering
    \includegraphics[width=0.9\linewidth]{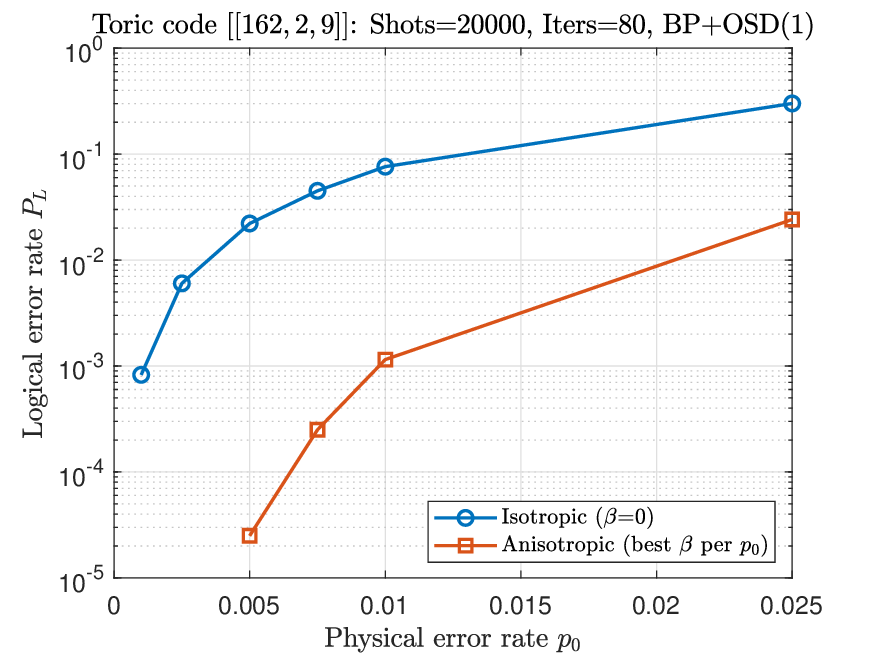}
    \caption{Logical error rate vs.\ physical error rate for the $\llbracket 162,2,9\rrbracket$ Toric code}
    \label{fig:toric_PL_p0}
\end{figure}

\begin{figure}
    \centering
    \includegraphics[width=0.9\linewidth]{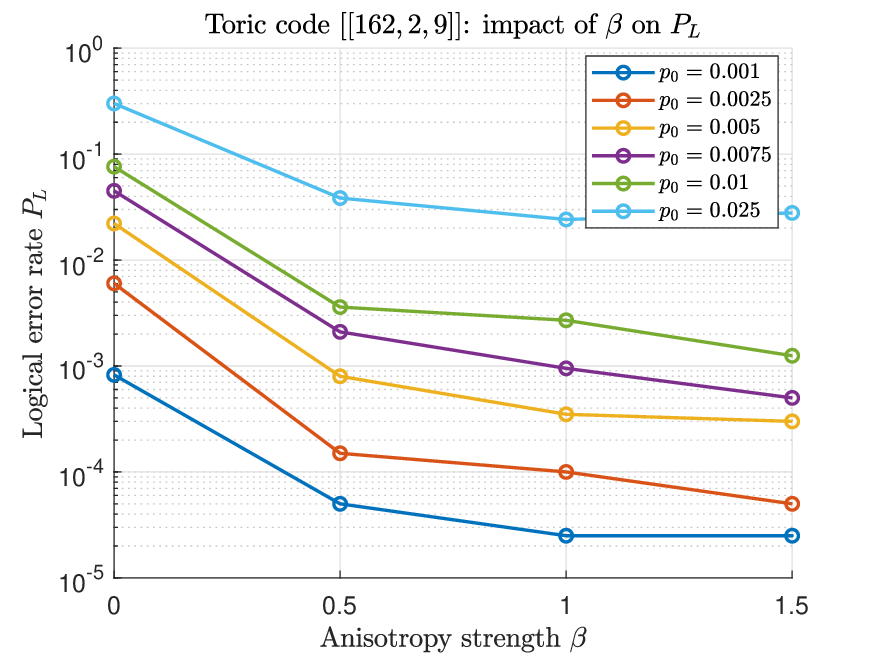}
    \caption{$P_L$ versus anisotropy strength $\beta$ for the $\llbracket 162,2,9\rrbracket$ Toric code}
    \label{fig:toric_PL_beta}
\end{figure}

\begin{figure}
    \centering
    \includegraphics[width=0.9\linewidth]{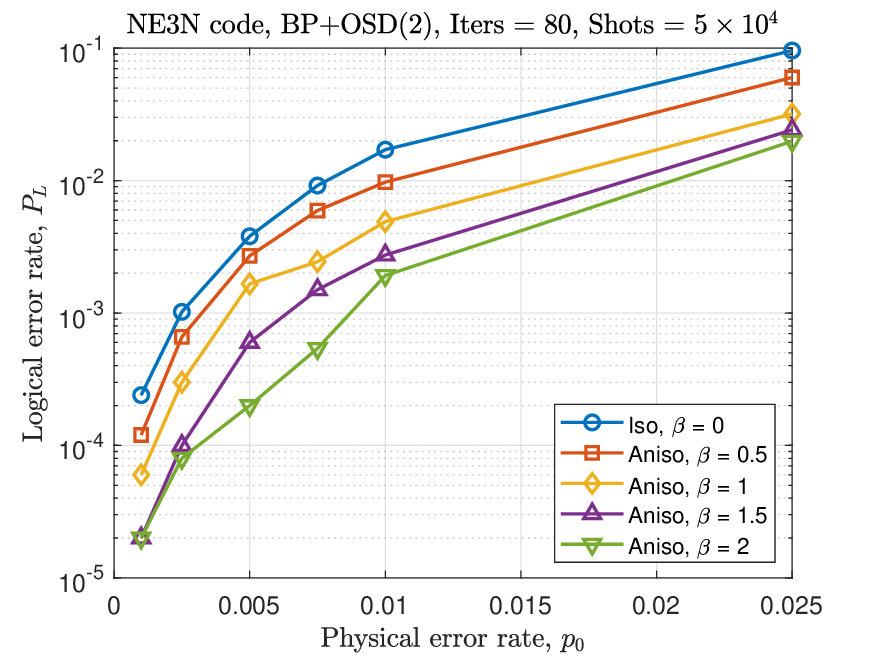}
    \caption{Logical error rate vs.\ physical error rate for the $\llbracket 36,4\rrbracket$ NE3N code} 
    \label{fig:ne3n_PL_p0}
\end{figure}

Fig.~\ref{fig:toric_PL_p0} shows that across the range $p_0\approx 10^{-3}$--$10^{-2}$, the directional prior, aligned with a geometric orientation of the toric lattice, reduces the logical error rate by roughly one to two orders of magnitude without modifying the code or the decoder architecture. 
Fig.~\ref{fig:toric_PL_beta} illustrates how the logical error rate varies with the anisotropy strength~$\beta$, treating the isotropic decoder
as the $\beta=0$ baseline. The trend confirms that the parameter~$\beta$ provides an effective and tunable control for enhancing BP+OSD decoding performance. 
For the $\llbracket 36,4\rrbracket$ NE3N code, Fig.~\ref{fig:ne3n_PL_p0} shows that orientation-based priors systematically outperform isotropic BP+OSD(2), reducing the logical error rate by roughly one order of magnitude across the tested physical error rates.

\section{Conclusion}
\label{sec:conclusion}
We introduced a directional framework for quantum LDPC decoding that annotates Tanner graphs with orientation-aware edge weights, aggregates them into per-qubit directional weights, and feeds these into both a directional degeneracy enumerator and an anisotropic BP+OSD decoder. The theory relates directional and Hamming distances, bounds the number of degenerate error classes per syndrome in terms of rate, distance, and directional bias, and admits a MacWilliams-type representation of the enumerator, while finite-length simulations show that modest anisotropy, controlled by a single parameter $\beta$, can already yield substantial logical error-rate reductions without changing the code or decoder architecture. The framework is deliberately simple and portable: per-qubit weights can be derived from geometry and scheduling, then passed through the same LLR mapping and BP+OSD pipeline, and the same machinery can be layered on top of Pauli-biased noise by starting from different baseline $X$- and $Z$-error rates and applying the spatial tilt separately to each. Extending the numerics to the larger families of codes and scenarios, incorporating circuit-level and correlated noise, and using the enumerator’s gradient identities to learn directional priors from data are natural directions for further work.


\bibliographystyle{IEEEtran}
\bibliography{refs}

\newpage
\appendix
\section*{Proof of Theorem \ref{thm:map-tree}}
\begin{proof}
Because errors are independent,
\[
\Pr(E) = \prod_{i=1}^n p_i^{E_i}(1-p_i)^{1-E_i}.
\]
The syndrome is a deterministic function $s=H E$, so
$\Pr(s\mid E)=\mathbf{1}_{\{H E=s\}}$ and hence
\[
\Pr(E\mid s) \;\propto\; \mathbf{1}_{\{H E=s\}}\Pr(E),
\]
so maximizing $\Pr(E\mid s)$ over $\mathcal{S}(s)$ is equivalent to
maximizing $\Pr(E)$ over $\mathcal{S}(s)$.

Taking minus the logarithm and discarding the additive constant that does not depend on $E$ gives
\[
-\log \Pr(E) = \sum_{i=1}^n E_i \log\frac{p_i}{1-p_i} + \mathrm{const}
= \sum_{i=1}^n w_i E_i + \mathrm{const},
\]
with $w_i=\log\frac{1-p_i}{p_i}$. Thus
\[
\arg\max_{E\in\mathcal{S}(s)} \Pr(E\mid s)
=
\arg\min_{E\in\mathcal{S}(s)} \sum_{i=1}^n w_i E_i,
\]
which proves the weighted min-sum characterization of the MAP estimate.

To see that min-sum BP produces this minimizer, represent the posterior as a factor graph
with variable nodes $E_i$ and check nodes enforcing the parity constraints $H_a E = s_a$.
The cost function is
\[
F(E) = \sum_{i=1}^n w_i E_i + \sum_{a=1}^m \Psi_a(E_{N(a)}),
\]
where $\Psi_a$ is $0$ if the parity at check $a$ is satisfied and $+\infty$ otherwise.
By assumption, the factor graph is a forest, i.e., a disjoint union of trees.
On a tree factor graph, min-sum (equivalently, max-product in the log domain)
belief propagation computes the exact global minimizer of any such additive,
locally factored cost~\cite{kschischang2002factor}. Applied to $F(E)$, min-sum BP
therefore returns the unique minimizer in $\mathcal{S}(s)$, which is the MAP estimate $E^\star$.
\end{proof}

\end{document}